\documentclass{amsart}

\RequirePackage[OT1]{fontenc}
\RequirePackage{amsthm,amsmath}
\RequirePackage[numbers]{natbib}
\RequirePackage[colorlinks,citecolor=blue,urlcolor=blue]{hyperref}
\usepackage{graphicx}
\usepackage{latexsym,color}
\usepackage{amssymb}
 \usepackage{amsfonts}
 \usepackage{amsmath}
\usepackage{amsthm}
\usepackage[outdir=./]{epstopdf}
\usepackage{stmaryrd}
\usepackage{enumitem}
\newtheorem{thm}{Theorem}[section]
 
 \newtheorem{lem}[thm]{Lemma}

\newtheorem{asm}[thm]{Assumption}

\numberwithin{equation}{section}

\newcommand{\Lp}{\left(}
\newcommand{\Rp}{\right)}

\newcommand{\Lc}{\left[}
\newcommand{\Rc}{\right]}

\newcommand{\IE}{\mathbb{E}}

\newcommand{\1}{\textbf{1}}

 \newcommand{\IR}{\mathbb{R}}

 \newcommand{\cS}{\mathcal{S}}

\def\Objectif{
\vspace{1.2cm}
$\diamond$ \underline{\sl Objectif :}
\advance\leftskip by 1.5cm
\vspace{1.2cm} \noindent }

\numberwithin{equation}{section}
\theoremstyle{plain}

\begin{document}

\title{Recombining tree approximations for Game Options in Local Volatility models}
\author{Benjamin Gottesman Berdah\\
 Hebrew University
}%
\address{
 Department of Statistics, Hebrew University of Jerusalem\\
 {e.mail: beni.gottesman@gmail.com}}

\date{\today}

\begin{abstract}
In this paper we introduce a numerical method for optimal stopping in the framework of one
dimensional diffusion. We use the Skorokhod embedding in order to construct
recombining tree approximations for diffusions with general coefficients. This technique allows us to
determine convergence rates and construct nearly optimal stopping times which are optimal at the
same rate. Finally, we demonstrate the efficiency of our scheme with several examples of game options.
\end{abstract}


\maketitle

\markboth{B. Gottesman Berdah}
{Recombining tree approximations for Game Options in Local Volatility models}
\renewcommand{\theequation}{\arabic{section}.\arabic{equation}}
\pagenumbering{arabic}


\section{Introduction}
Game contingent claim (GCC) or game option, which have been introduced by Kifer \cite{Ki1}, is a derivative contract between the seller and the buyer of the option, where both have the right to exercise it at any time before a maturity date $T$. The buyer pays an initial amount, which correspond to the price of the option. 
If the buyer exercises the contract at time $t$, before the seller cancels, then he receives the payment $Y_t$, but if the seller cancels the contract before the buyer then the latter receives $X_t$. 
The difference $\Delta_t=X_t-Y_t$ is called the penalty which the seller needs to pay to the buyer for the contract cancellation.
Concretely, if the seller will exercise at a stopping time $\sigma\leq{T}$ and the buyer at a stopping time $\tau\leq{T}$
then the former pays to the latter the amount $H(\sigma,\tau)$ where
\begin{equation}\label{3.1.1}
H(\gamma,\tau):=X_{\gamma}\mathbb{I}_{\gamma<\tau}+Y_{\tau}\mathbb{I}_{\tau\leq{\gamma}}
\end{equation}
with $\mathbb{I}_{Q}=1$ if an event $Q$ occurs and $\mathbb{I}_{Q}=0$ if not. 
Without loss of generality we assume that the payoff $H(\gamma,\tau)$ is discounted.

Consider a local volatility model with time horizon $T<\infty$,
which
consists of a riskless savings account with constant interest rate
and of a risky asset whose discounted value at time $t$ satisfies the stochastic differential equation (SDE)
\begin{equation}\label{3.1.2}
\frac{dS_t}{S_t}=\sigma(S_t)dW_t
\end{equation}
with a given initial value $S_0>0$.

From the martingale representation theorem it follows (under some regularity assumptions) that the model which is given by (\ref{3.1.2}) is a model of a complete
market (for details see Section 1.6 in \cite{newKS}). Moreover, 
without loss of generality we assume that 
the market measure is the unique martingale measure.

There are several papers which dealt with the computations of game options prices (see, for instance \cite{newK2,DG2018,E,newK1,KKnew,KKS,YYZ}),
however as far as I know only \cite{KKnew,YYZ} dealt with finite maturity game options. These two papers considered the (constant volatility)
Black--Scholes model.

It is well known that
pricing game options (see \cite{Ki2} and the references there) leads to Dynkin games. For finite maturity Dynkin games
there are no explicit solutions even in the relatively simple framework where the
diffusion process is a standard Brownian motion, and so it is important to obtain efficient numerical schemes.

In this article
we extend the results from \cite{BDG2018} which studied numerical schemes for American options in local volatility models.
We construct recombing tree approximations
for Dynkin games in local volatility models and we obtain the same error estimates as in \cite{BDG2018}.
Namely, our method allows to compute the corresponding value and the optimal control with complexity $O(n^2)$ and
error estimates of order $O(n^{-1/4})$,
where $n$ is the number of time steps.
Finally, we apply our technique and provide several numerical results.

\section{Preliminaries and the Main Result}
Consider a complete probability space
$(\Omega, \mathcal{F}, \mathbb P)$ together with a standard
one--dimensional Brownian motion
$\{W_t\}_{t=0}^\infty$, and the filtration
$\mathcal{F}_t=\sigma{\{W_s|s\leq{t}\}}$ completed by the null sets.

We consider the model given by (\ref{3.1.2}).
Set $Z_t:=\ln S_t$. From the It\^o's formula
\begin{equation}\label{3.2.1-}
dZ_t=\psi(Z_t) dW_t-\frac{\psi^2(Z_t)}{2} dt, \ \ Z_0=\ln S_0.
\end{equation}
where
$\psi(z):=\sigma(e^z)$, $z\in\mathbb R$.

\begin{asm}\label{BoundedAssumption}
We assume that $\psi:\mathbb R\rightarrow\mathbb R_{+}$ is a Lipschitz continuous function such that
$\psi,\frac{1}{\psi}$ are bounded.
\end{asm}

In particular, assumption (\ref{BoundedAssumption}) implies that the
SDE (\ref{3.2.1-}) has a \emph{unique strong solution}.
Moreover, since $\frac{1}{\sigma}$ is uniformly bounded then
the market model which is given by (\ref{3.1.2}) is complete.
Thus, from \cite{Ki1,Ki2} we obtain that the price of the game contingent claim given by
(\ref{3.1.1})
 equals to
\begin{align}\label{fairDG}
V:=\inf_{\gamma\in\mathcal T_T}\sup_{\tau\in\mathcal T_T}\mathbb E\left[H(\gamma,\tau)\right]=
\sup_{\tau\in\mathcal T_T}\inf_{\gamma\in\mathcal T_T}\mathbb E\left[H(\gamma,\tau)\right].
\end{align}

Assume that the payoffs are given by $X_t:=g(t,S_t)$, $Y_t:=f(t,S_t)$
where $g,f:[0,T]\times\mathbb R_{+}\rightarrow \mathbb R$ satisfy
$g\geq f$ and, the following Lipschitz condition
\begin{align*}
&\mid f(t_1,x_1)-f(t_2,x_2)\mid +\mid g(t_1,x_1)-g(t_2,x_2)\mid +\mid h(t_1,x_1)-h(t_2,x_2)\mid\nonumber \\
&\leq L \Lp  (1+\mid x_1\mid )\mid t_2-t_1\mid +\mid x_2-x_1\mid \Rp , \ t_1,t_2\in [0,T], \ x_1,x_2\in \mathbb R_{+} \label{2.function}
\end{align*}
for some constant $L$.

We aim to approximate efficiently the value $V$. As in Section 2.3 our
main tool will be the Skorokhod embedding technique.
\subsection{Skorokhod embedding}
Fix $n\in\mathbb N$ and denote $h:=\frac{T}{n}$. Set
$\overline{\sigma}:=\sup_{z\in\mathbb R}\psi(z)$.
 We want to construct a sequence of stopping times (on the Brownian probability space)
$0<\theta^{(n)}_1<...<\theta^{(n)}_n$ such that for any $k$
\begin{equation}\label{3.2.1}
Z_{\theta^{(n)}_{k+1}}-Z_{\theta^{(n)}_k}\in
\{-\overline{\sigma}\sqrt h,0,\overline{\sigma}\sqrt{ h}
\}
\end{equation}
 and
\begin{equation}\label{3.2.2}
\IE(\theta^{(n)}_{k+1}-\theta^{(n)}_k\mid \mathcal F_{\theta^{(n)}_k})=h+O(n^{-3/2}).
\end{equation}

To that end, we apply the results from Section 2.3 in \cite{BDG2018}.
For any
$A\in [0, \overline{\sigma}\sqrt h]$ consider the stopping times
$$\rho^{Z_0}_A=\inf\{t: |Z_t-Z_0|=A\}$$
and
$$\kappa^{Z_0}_A=\sum_{i=1}^2  \mathbb I_{Z_{\rho^{Y_0}_A}=Z_0+(-1)^i A}
\inf\{t\geq \rho^{Z_0}_A: Z_t=Z_0 \ \mbox{or} \ Z_t=Z_0+(-1)^i \overline{\sigma}\sqrt{h}\}.
$$
\begin{lem}
Define the stopping times
$\theta^{(n)}_1,...,\theta^{(n)}_n$ by the following recursive relations
\begin{equation*}\label{ThetaDef}
\left\{
\begin{aligned}
\theta^{(n)}_0&:=0\\
\theta^{(n)}_{k}&:= \kappa^{Z_{\theta^{(n)}_{k-1}}}_{\sigma^2\Lp Z_{\theta^{(n)}_{k-1}}\Rp\sqrt{h}/\overline{\sigma}},\
\text{for}\ k=1,...,n.
\end{aligned}
\right .
\end{equation*}
Then the stopping times $\theta^{(n)}_0,...,\theta^{(n)}_n$ satisfy (\ref{3.2.1})--(\ref{3.2.2}).
\end{lem}
\begin{proof}
The proof follows from Section 2.3 in \cite{BDG2018} (see Remark 2.2 there).
\end{proof}
Next, introduce the functions
$$p^{(1)}(z):=\frac{
\left(1-e^{-\sigma^2 z \sqrt{h}/\overline{\sigma}}\right) \left(e^{\sigma^2 z \sqrt{h}/\overline{\sigma}}-1\right)}
{\left(e^{\sigma^2 z \sqrt{h}/\overline{\sigma}}-e^{-\sigma^2 z \sqrt{h}/\overline{\sigma}}\right)
\left(e^{\overline{\sigma}\sqrt{h}}-1\right)}, \ \ z\in\mathbb R,$$
$$p^{(-1)}(z)=e^{\overline{\sigma}\sqrt{h}}p^{(1)}(z),  \ \ z\in\mathbb R$$
and
$$p^{(0)}(z)=1-p^{(1)}(z)-p^{(-1)}(z), \ \ z\in\mathbb R.$$

Observe
that the support of the random variable $Z_{\rho^{Z_0}_A}$ and
$Z_{\kappa^{Z_0}_A}-Z_{\rho^{Z_0}_A}|Z_{\rho^{Z_0}_A}$ consist of only two points.
Thus,
from the strong Markov property
of $Z$ and the fact that $e^Z$
is a martingale we obtain
\begin{equation*}
\mathbb P\left(Z_{\theta^{(n)}_{1}}=Z_0+(-1)^{i}\overline{\sigma} \sqrt{h}\right)=
p^{(i)}(Z_0), \ \ i=-1,0,1.
\end{equation*}
By applying again the strong Markov property we conclude that for any $k$
\begin{equation}\label{3.2.3}
\mathbb P\left(Z_{\theta^{(n)}_{k+1}}=Z_{\theta^{(n)}_{k}}+(-1)^{i}\overline{\sigma} \sqrt{h}|\mathcal F_{\theta^{(n)}_k}\right)=
p^{(i)}(Z_k),  \ \ i=-1,0,1.
\end{equation}

\subsection{Dynkin Games for Trinomial Models}
For a given $n$, denote by $\cS_n$ the set of all stopping time with respect to the filtration
$\{\mathcal F_{\theta^{(n)}_k}\}_{k=0}^n$, with values in the set $\{0,1,...,n\}$.
Introduce the Dynkin game value
\begin{equation}\label{fairDGn}
\begin{aligned}
V_n:&=\inf_{\zeta\in\cS_n}\sup_{\eta\in\cS_n}
\IE \left[ g\left(\zeta h, S_{\theta^{(n)}_\zeta}\right)\mathbb I_{\zeta<\eta} + f\left(\eta h, S_{\theta^{(n)}_\eta}\right)\mathbb I_{\eta\leq \zeta}\right]\\
&=\sup_{\eta\in\cS_n}\inf_{\zeta\in\cS_n}
\IE \left[ g\left(\zeta h, S_{\theta^{(n)}_\zeta}\right)\mathbb I_{\zeta<\eta} + f\left(\eta h, S_{\theta^{(n)}_\eta}\right)\mathbb I_{\eta\leq \zeta}\right].
\end{aligned}
\end{equation}

Recall that $S_t:=e^{Z_t}$. Hence,
the process $\{S_{\theta^{(n)}_k}\}_{k=0}^n$ lies on the grid
$S_0\exp\Lp  \overline{\sigma}\sqrt{h} i\Rp $, $i=-n,1-n,...,0,1,...,n$.

By combining standard dynamical programming for Dynkin games (see \cite{O}),the strong Markov property
of $S$ and the transition probabilities given by (\ref{3.2.3}) we compute $V_n$
by the following backward recursion.

Define the functions
\begin{flalign*}
J^{(n)}_k : \{Z_0+\overline{\sigma}\sqrt h\{-k, 1-k, ... &, 0, 1, ..., k \} \}\rightarrow \IR,\ \ k=0,...,n\\
J^{(n)}_n(z)&=f(T,e^z)
\end{flalign*}
and for $k=0,1,...,n-1$
$$
J^{(n)}_k (z) = \max \left( f(kh,e^z), \min \left(g(kh,e^z), \sum_{i=-1,0,1}p^{(i)}J^{(n)}_{k+1} (z+i\overline{\sigma}\sqrt h)\right)\right).
$$
We get that
$$
V_n = J^{(n)}_0 (Z_0).
$$
Moreover, the stopping time given by
$$
\eta^{*}_n := n\wedge\min\left\{k:J^{(n)}_k (Z_{\theta^{(n)}_k}) = f\left(kh, e^{Z_{\theta^{(n)}_k}}\right)\right\}
$$
is an optimal stopping time for the buyer and
$$
\zeta^{*}_n := n\wedge \min\left\{k:J^{(n)}_k (Z_{\theta^{(n)}_k}) = g\left(kh, e^{Z_{\theta^{(n)}_k}}\right)\right\}
$$
is an optimal stopping time for the seller.

Namely,
\begin{equation}\label{fairDGn1}
\begin{aligned}
V_n:&=\sup_{\eta\in\cS_n}
\IE \left[ g\left(\zeta^{*}_n h, S_{\theta^{(n)}_{\zeta^{*}_n}}\right)\mathbb I_{\zeta^{*}_n<\eta} + f\left(\eta h, S_{\theta^{(n)}_\eta}\right)\mathbb I_{\eta\leq \zeta^{*}_n}\right]\\
&=\inf_{\zeta\in\cS_n}
\IE \left[ g\left(\zeta h, S_{\theta^{(n)}_\zeta}\right)\mathbb I_{\zeta<\eta^{*}_n} + f\left(\eta^{*}_n h, S_{\theta^{(n)}_{\eta^{*}_n}}\right)
\mathbb I_{\eta^{*}_n\leq \zeta}\right].
\end{aligned}
\end{equation}

As in Section 2.2 the grid structure allows to compute $V_n$ with complexity $O(n^2)$.
We arrive to the approximation result.
\begin{thm}\label{thm2.2}
The values $V$ and $V_n$ defined respectively by (\ref{fairDG}) and (\ref{fairDGn}) satisfy
\begin{equation*}\label{-1}
\mid V-V_n \mid = O (n^{-1/4}).
\end{equation*}
Moreover,
for the stopping times
$\tau^*_n:=T\wedge\theta^{(n)}_{\eta^{*}_n}$
and
$\gamma^*_n:=T\mathbb {I}_{\zeta^{*}_n=n}+(T\wedge \theta^{(n)}_{\zeta^{*}_n})\mathbb I_{\zeta^{*}_n<n}$
we have
\begin{equation*}\label{-2}
V-\inf_{\gamma\in\mathcal T}\IE\Lc  H(\gamma,\tau^{*}_n) \Rc = O (n^{-1/4})
\end{equation*}
and
\begin{equation*}\label{-3}
\sup_{\tau\in\mathcal T}\IE\Lc H(\gamma^*_n,\tau) \Rc-V = O (n^{-1/4}).
\end{equation*}
\end{thm}
\begin{proof}
Fix $n\in\mathbb N$ and denote $h:=\frac{T}{n}$. Let $\tau\in\mathcal T_T$. Define the map $\varphi_n:\mathcal T_T\rightarrow\mathcal S_n$
by
$$\varphi_n(\tau)=n\wedge\min\{k:\theta^{(n)}_k\geq \tau\}, \ \ \tau\in\mathcal T_T.$$
Observe that for any $\tau\in\mathcal T_T$ we have
$$|\tau-\theta^{(n)}_{\varphi_n(\tau)}|\leq |T-\theta^{(n)}_n|+\max_{1\leq k\leq n}\theta_k-\theta_{k-1} \leq h+3\max_{0\leq k\leq n}|\theta^{(n)}_k- kh|.$$
Thus, by applying (\ref{3.2.1})--(\ref{3.2.2})
and using the exactly the same arguments as in Section 2.4 we obtain
\begin{equation}\label{-10}
\sup_{\tau\in\mathcal T_T}\mathbb E_{\mathbb P}\left[\left|
f\left(\varphi_n(\tau) h, S_{\theta^{(n)}_{\varphi_n(\tau)}}\right)-f\left(\tau,S_{\tau}\right)\right|\right]= O(n^{-1/4}).
\end{equation}
Next, we notice that
$$|\zeta^{*}_n h-\gamma^{*}_n|\leq h+\max_{0\leq k\leq n}|\theta^{(n)}_k- kh|.$$
Thus, (again we use the same arguments as in Section 2.4 )
\begin{equation}\label{-9}
\mathbb E_{\mathbb P}\left[\left|
g\left(\zeta^{*}_n h, S_{\theta^{(n)}_{\zeta^{*}_n}}\right)-g\left(\gamma^{*}_n,S_{\gamma^{*}_n}\right)\right|\right]=O(n^{-1/4}).
\end{equation}
From the definitions it is clear that for a given stopping time $\tau\in\mathcal T_T$
the inequality
$\gamma^{*}_n<\tau$ implies $\zeta^{*}_n<\varphi_n(\tau)$. Namely,
$$\{\gamma^{*}_n<\tau\}\subset \{\zeta^{*}_n<\varphi_n(\tau)\}.$$
Thus, from (\ref{fairDGn1})--(\ref{-9})
\begin{eqnarray}
&\sup_{\tau\in\mathcal T}\mathbb E_{\mathbb P}[H(\gamma^*_n,\tau)]-V_n\leq \nonumber\\
&\sup_{\tau\in\mathcal T_T}\mathbb E_{\mathbb P}\left[\left|
f\left(\varphi_n(\tau) h, S_{\theta^{(n)}_{\varphi_n(\tau)}}\right)-f\left(\tau,S_{\tau}\right)\right|\right]\nonumber\\
&+\mathbb E_{\mathbb P}\left[\left|
g\left(\zeta^{*}_n h, S_{\theta^{(n)}_{\zeta^{*}_n}}\right)-g\left(\gamma^{*}_n,S_{\gamma^{*}_n}\right)\right|\right]=O(n^{-1/4}).\label{-8}
\end{eqnarray}

Next, define the map
$\tilde\varphi_n(\gamma):\mathcal T_T\rightarrow\mathcal S_n$ by
$$\tilde\varphi_n(\gamma)=\left(n\wedge\min\{k:\theta^{(n)}_k\geq \gamma\}\right)\mathbb{I}_{\gamma<T}+n\mathbb I_{\gamma=T}, \ \ \gamma\in\mathcal T_T.$$
Notice that for any $\gamma\in\mathcal T_T$ we have
$$|\gamma-\theta^{(n)}_{\tilde\varphi_n(\gamma)}|\leq |T-\theta^{(n)}_n|+\max_{1\leq k\leq n}\theta_k-\theta_{k-1} \leq h+3\max_{0\leq k\leq n}|\theta^{(n)}_k- kh|.$$

Moreover, we observe that for any stopping time $\gamma\in\mathcal T_T$
$$\{\tilde\varphi_n(\gamma)<\eta^{*}_n\}\subset \{\gamma<\tau^{*}_n\}.$$
Thus, from (\ref{fairDGn1})
\begin{eqnarray}
&V_n-\inf_{\gamma\in\mathcal T}\mathbb E_{\mathbb P}[H(\gamma,\tau^{*}_n)]\leq \nonumber\\
&\sup_{\gamma\in\mathcal T_T}\mathbb E_{\mathbb P}\left[\left|
g\left(\tilde\varphi_n(\gamma) h, S_{\theta^{(n)}_{\tilde\varphi_n(\gamma)}}\right)-g\left(\gamma,S_{\gamma}\right)\right|\right]\nonumber\\
&+\mathbb E_{\mathbb P}\left[\left|
f\left(\eta^{*}_n h, S_{\theta^{(n)}_{\eta^{*}_n}}\right)-f\left(\tau^{*}_n,S_{\tau^{*}_n}\right)\right|\right]=O(n^{-1/4})\label{-7}
\end{eqnarray}
where the estimate is obtained exactly as in (\ref{-8}).

From (\ref{-8})--(\ref{-7})
$$V_n-O(n^{-1/4})\leq\inf_{\gamma\in\mathcal T}\mathbb E_{\mathbb P}[H(\gamma,\tau^{*}_n)] \leq V\leq\sup_{\tau\in\mathcal T}\mathbb E_{\mathbb P}[H(\gamma^*_n,\tau)]\leq V_n+O(n^{-1/4})$$
and the proof is completed.
\end{proof}

\section{Numerical Results}
Consider the local volatility model which is given by
\begin{equation*}
\frac{dS_t}{S_t}=\min\left(0.5,\max\left(0.05,\frac{\sqrt S_t}{30}\right)\right)dW_t.
\end{equation*}
This model can be viewed as a truncated version of the CEV model (\cite{CEV}).
The process $S_t$, $t\geq 0$ denotes the discounted stock price.
We assume that we have a constant interest rate $r=0.06$.

\subsection{Game Call Options}
Consider a game call option with strike price $K=100$ and constant penalty $\delta=12$.
Namely, the discounted payoff is given by
$$H(\gamma,\tau):=e^{-0.06 (\gamma\wedge\tau)}\left((S_{\gamma\wedge\tau}-100)^{+}+12
\mathbb{I}_{\gamma<\tau}\right).$$
We assume that the maturity date is $T=2$.

First, by applying the constructed above trinomial trees,
we compute (Table (\ref{table31})) the option prices for different initial stock prices.
\begin{table}\label{table31}
\begin{center}
\caption{
We provide numerical results for game call options with the above parameters, with different initial
stock prices. The number of steps in the trinomial approximations denoted by $n$.}
\begin{tabular}{lcccc}
\hline
\multicolumn{5}{c}{Game Call Option Prices}\\
\cline{2-5}
$S_0$ & n = 400 & n = 700 & n = 1200 & n = 2000  \\
\hline
80 & 6.8637 & 6.8357 & 6.8081 & 6.7823\\
85 & 8.2609 & 8.2221 &  8.1884  & 8.1433  \\
90 & 9.6056 & 9.5534 & 9.5407 & 9.5083 \\
95 & 10.9539 & 10.9332 & 10.9123 & 10.8943  \\
100 & 12 & 12 &12 & 12  \\
105 & 17 & 17 &17 & 17  \\
110 & 22 & 22 & 22 & 22 \\
\hline
\end{tabular}
\end{center}
\end{table}

Next, we calculate numerically the stopping regions.
For American call options the discounted payoff is a sub--martingale (under the martingale measure)
and so the buyer stopping time is $\tau\equiv T$.

It remains to treat the seller. Set,
$$
V^{call}(u,x):=\inf_{\gamma\in\mathcal T_u}\sup_{\tau\in\mathcal T_u}\mathbb E\left[e^{-0.06 (\gamma\wedge\tau)}\left((S_{\gamma\wedge\tau}-100)^{+}+12
\mathbb{I}_{\gamma<\tau}\right)\right], \ \ u,x>0
$$
where $S_0=x$.

We observe that the optimal stopping time for the seller is given by (recall that $S$ is the discounted stock price)
$$
\gamma^{*}=T\wedge\inf\{t: e^{0.06 t}S_t\in D^{call}\}
$$
where
$$
D^{call}=\{(t,x):  V^{call}(T-t,x)=(x-100)^{+}+12\}.
$$

We obtain numerically (Figure \ref{fi:CallStoppingRegionsLocalVol}) that the structure of the stopping region $D^{call}$ is of the form
$$D=\{(t,x): t\in [0,\mathcal T_1], \ K\leq x\leq b^{call}(t)\}\bigcup \left\{[\mathcal T_1,\mathcal T_2]\times\{K\}\right\}$$
where $\mathcal T_1<\mathcal T_2<T$ and $b^{call}:[0,\mathcal T_1]\rightarrow [K,\infty)$.

\subsection{Game Put Options}
We consider a game put option with strike price $K=100$ and constant penalty $\delta=12$.
Thus, the discounted payoff is given by

$$
H(\gamma,\tau):=e^{-0.06 (\gamma\wedge\tau)}\left((100-S_{\gamma\wedge\tau})^{+}+12
\mathbb{I}_{\gamma<\tau}\right).
$$ 

As before we take the maturity $T=2$.

\begin{figure}
\includegraphics[width=0.9\textwidth]{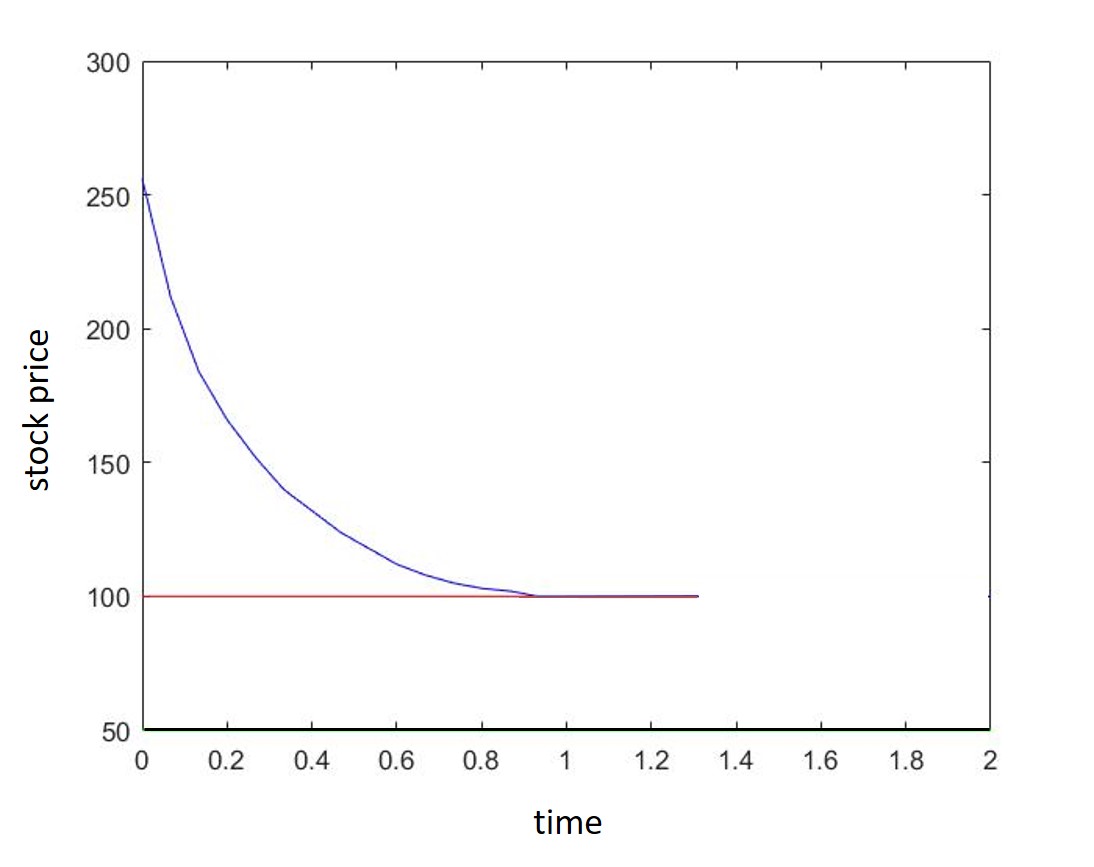}
\caption{We consider a game call option with
the above parameters.
We take $n=2000$ and compute numerically the stopping region for the seller. We get that for $t\in [0,0.93]$ the seller should exercise at the first moment when the stock price is between the strike price and the
value given by the blue curve. For $t\in [0.93,1.33]$ the seller stops at the first moment the stock price equals to the strike price.
After the time $t=1.33$ the investor should not exercise (before the maturity date).}
\label{fi:CallStoppingRegionsLocalVol}
\end{figure}

\begin{table}\label{table32}
\begin{center}
\caption{
We provide numerical results for game put options with the above parameters, with different initial
stock prices. The number of steps in the trinomial approximations denoted by $n$.}
\begin{tabular}{lcccc}
\hline
\multicolumn{5}{c}{Game Put Option Prices}\\
\cline{2-5}
$S_0$ & n = 400 & n = 700 & n = 1200 & n = 2000  \\
\hline
80 & 22.6243 & 22.6312 & 22.6341 & 22.6184\\
85 & 19.7150 & 19.6465 &  19.6027  & 19.5848  \\
90 & 16.9593 & 16.9420 & 16.9110 & 16.8969 \\
95 & 14.4512 & 10.4282 & 14.4104 & 14.3933  \\
100 & 12 & 12 &12 & 12  \\
105 & 11.1356 & 11.0930 &11.0841 & 11.0378  \\
110 & 10.1854 & 10.1368 & 10.0905 & 10.0385 \\
115 & 9.1742   &   9.1132 &   9.0713 & 9.0622 \\
120 & 8.3025 & 8.2646 & 8.2529 & 8.2378 \\
\hline
\end{tabular}
\end{center}
\end{table}

In Table (\ref{table32}) we compute the option prices for different initial stock prices.

Finally, we calculate numerically the stopping regions.
We start with the seller. In \cite{KKnew} the authors showed that for finite horizon game put options in the Black--Scholes model the stopping time for the seller is of the form
$$\gamma^{*}=T\wedge\inf\{t\in [0,t^{*}]:e^{rt} S_t=K\}$$
for some $t^{*}$ which the authors characterize. In fact, their arguments are valid for any local volatility model. Of course, the characterization of $t^{*}$ is more complicated
in models with non constant parameters.
\begin{figure}
\includegraphics[width=0.9\textwidth]{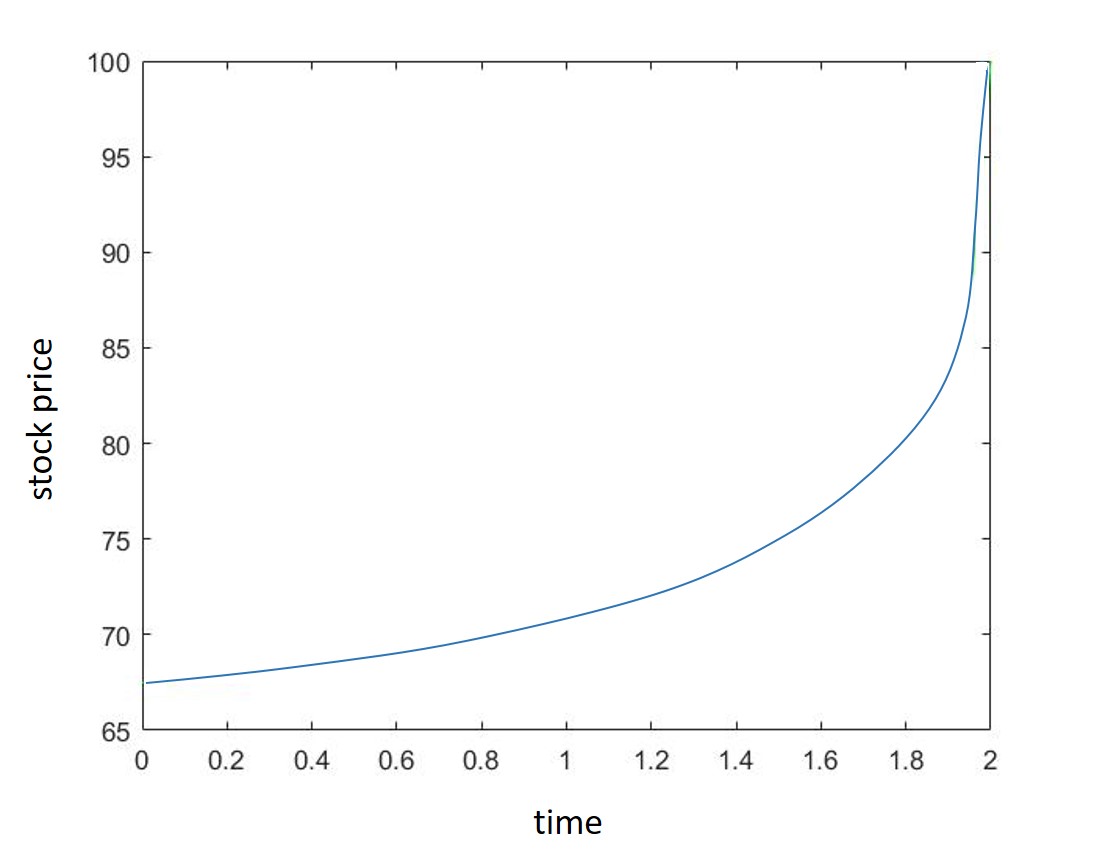}
\caption{We consider a game put option with
the above parameters.
We take $n=2000$ and compute numerically the stopping regions for the buyer.
We get that the holder should exercise at the first moment when the
stock price is below the
value given by the blue curve.}
\label{fi:PutStoppingRegionsLocalVol}
\end{figure}
By applying our trinomial models we show numerically that the seller stopping time
is given by
$$\gamma^{*}=T\wedge\inf\{t\in [0,1.33]:e^{rt} S_t=100\}.$$

Namely, after time $t^{*}=1.33$ the seller wait for the maturity date.
Observe that for game call options (with the same parameters) we also obtained numerically that after time
$1.33$ the seller wait for the maturity date.
An open question, which we leave for future research is to understand whether
this is just a coincidence or whether there is some connection between game call options and game put options ?

It remains to compute numerically the holder stopping time.
Roughly speaking, the holder will reason in the same way as
he would for the associated American put. That is to make a compromise between the stock
reaching a prescribed low value and not waiting too long.

Thus, we expect that the holder stopping time will of the form
$$\tau^{*}=T\wedge\inf\{t:e^{rt} S_t\in D^{put}\}
$$
where
the stopping region $D^{put}$ is of the form
$$D=\left\{(t,x): t\in [0,2], \ x\leq \phi(2-t)\right\}$$
where $\phi:[0,2]\rightarrow (0,100]$ is a decreasing function.
In Figure \ref{fi:PutStoppingRegionsLocalVol} we confirm this numerically.

\section*{Acknowledgments}
I would like to cordially thank my adviser and teacher, Yan Dolinsky, for guiding me and, for the long and fruitful discussions on this topic.

\end{document}